\documentclass[12pt,a4paper,final]{article}
\usepackage[utf8]{inputenc}
\usepackage{amsmath,amsthm}
\usepackage{amsfonts}
\usepackage{amssymb}
\usepackage{makeidx}
\usepackage{graphicx}
\usepackage{theoremref}
\usepackage{enumitem}

\newcommand{\Fq}{\mathbb{F}_{q}}
\newcommand{\Fp}{\mathbb{F}_{p}}
\newcommand{\Fpn}{\mathbb{F}_{p^n}}
\newcommand{\Fpm}{\mathbb{F}_{p^m}}
\newcommand{\sums}[1]{\sum_{\substack{#1}}}

\DeclareMathOperator*{\Tr}{Tr}

\newtheorem{theorem}{Theorem}
\newtheorem{proposition}{Proposition}

\newtheorem{lemma}{Lemma}
\newtheorem{definition}{Definition}
\newtheorem{remark}{Remark}
\begin{document}

\title{An updated review on crosscorrelation of m-sequences\footnote{This is an invited chapter to \textit{Series on Coding Theory and Cryptology, World Scientific}}}

\author{Tor Helleseth and Chunlei Li \\
\small Department of Informatics, \\
\small University of Bergen, Norway\\
\small  Email: tor.helleseth@uib.no, \, chunlei.li@uib.no 
}
\date{}
\maketitle

\begin{abstract}
Maximum-length sequences (m-sequences for short) over finite fields are generated by linear feedback shift registers with primitive characteristic polynomials.
These sequences have nice mathematical structures and good randomness properties that are favorable in practical applications. During the past five decades, the crosscorrelation
between m-sequences of the same period has been intensively studied, and a particular research focus has been on investigating the crosscorrelation spectra with few possibles values. 
In this chapter we summarize all known results on this topic in the literature and promote several open problems for future research.
\end{abstract}



\section{Introduction}

``\textit{An octillion. A billion billion billion. That's a fairly conservative estimate of the number of times a cellphone or other device somewhere in the world has generated a bit using a maximum-length linear-feedback shift register sequence. It's probably the single most-used mathematical algorithm idea in history}." Wolfram wrote this in the beginning of his article \cite{Wolfram2016} dedicated in memory of Solomon Golomb, who is the major originator of shift register sequences.
Wolfram's description reflects Golomb's earlier statement in his own book ``Shift Register Sequences" \cite{Gold1967} that ``the theory of shift register sequences has found major applications
in a wide variety of technological situations, including secure, reliable
and efficient communications, digital ranging and tracking systems,
deterministic simulation of random processes, and computer sequencing
and timing schemes." In communication systems, maximum-length linear-feedback sequences (m-sequences for short) have been mainly used to distinguish multiplexed signals at the receiver based on their nice correlation property.
One notable example is the \textit{Gold sequence} (also known as the \textit{Gold code}), which is a family of  binary sequences derived from m-sequences.
Gold sequences have bounded small cross-correlation, which is useful when multiple devices are broadcasting in the same frequency range, such as in code-division multiple access	(CDMA) systems and Global Positioning System (GPS) \cite{Helleseth-Kumar}.
The origins of wide applications of Gold sequences  go back to the pioneering work of Gold in \cite{Gold1968}.
The ideal autocorrelation property of m-sequences is easy to deduce and had been well known in the community earlier.
In 1968 Gold initiated\footnote{Gold pointed out that in the same time period Kasami had obtained similar results in the context of coding theory \cite{Kasami1966}}  the study of crosscorrelation of an m-sequence and its decimated sequence of the same period  and completely determined
each value of the crosscorrelation function and its occurrence. 
During more than 50 years, the crosscorrelation of m-sequences has been extensively studied \cite{Helleseth-Kumar,Helleseth2014a}, and researches on this topic 
have given rise to many interesting results and also revealed its close connections to important objects in other areas, including almost bent functions and almost perfect nonlinear functions in cryptography \cite{carlet2020}, cyclic codes in coding theory \cite{Dobbertin-Felke-Helleseth-Rosendahl}  and finite geometry \cite{Games-Sequences,Games-Quadrics}. There have been some overviews on the advances of crosscorrelation of m-sequences and relevant topics, e.g., \cite{Helleseth2002}, \cite{Helleseth2014a}, \cite{Katz2018}. 
In this chapter we will focus on few-valued crosscorrelation between m-sequences of the same period and provide an updated overview on the topic with new techniques and results in recent years.

The remainder of the chapter is organized as follows: Section 2 starts with the generation and properties of m-sequences, then briefly recalls the properties of crosscorrelation of m-sequences and their connections to the Walsh spectrum of cryptographic functions and the weight distribution of cyclic codes; Sections 3-6 review known decimations that lead to 3-valued, 4-valued, 5-valued and 6-valued crosscorrelations, respectively, where the main ideas and techniques are described and some open problems are given.
Section 7 concludes the chapter.

\section{Crosscorrelation of m-sequences}\label{ra_sec1}

Assume $p$ is a prime and let $\Fp$ denote the finite field with $p$ elements. 
Given an initial state $(s_0, s_1, \ldots, s_{n-1})$, the linear recursion of degree $n$ over $\Fp$ defined by
\begin{equation*}\label{key}
s_{t+n} + c_{n-1}s_{t+n-1} + \cdots + c_1s_{t+1} + c_0s_t =0, \quad c_0\neq 0,
\end{equation*} generates a periodic sequence $\{s_t\}$. The characteristic polynomial associated with the linear recursion is given by
$$
c(x)=x^n + c_{n-1}x^{n-1} + \cdots + c_1x + c_0.
$$
For a linear recursion of degree $n$ over $\Fp$, its output sequences  have maximal period $p^n-1$ since 
they are determined by the initial states and there are in total $p^n-1$ distinct $n$-tuples over $\Fp$ except for the all-zero tuple. When the characteristic polynomial $c(x)\in \Fp[x]$ is a primitive polynomial of degree $n$ (namely, the smallest integer $e$ such that $c(x)$ divides $x^e-1$ is $p^n-1$),
the recursion is known to generate maximum-length linear recursive sequences, or m-sequences for short \cite{Selmer1966,Golomb1967}. 

Suppose $\alpha$ is a root of the primitive polynomial $f(x)$ that generates an m-sequence $\{s_t\}$ over $\Fp$. 
Let $\Tr$ be the absolute trace function from $\Fpn$ to $\Fp$ defined by $\Tr(x) = x + x^p + \cdots + x^{p^{n-1}}$. 
With the one-to-one correspondence between the absolute trace functions and the linear functions from $\Fpn$ to $\Fp$, the m-sequence $\{s_t\}$ can (after a suitable cyclic shift) be expressed simply by 
$
s_t = \Tr(\alpha^t).
$ All nontrivial cyclic shifts of the m-sequence $\{s_t\}$ are given by 
$$
\{s_{t+\tau}\} \text{ with } s_{t+\tau} = \Tr(\alpha^{\tau}\alpha^t), \quad 0<\tau \leq p^n-2.
$$  These shifts are also m-sequences, and they are called the {\rm shift equivalence} of $\{s_t\}$.

Golomb  in \cite{Golomb1967} characterized some important properties of m-sequences.

\begin{lemma} \thlabel{Lem1} Let $f(x) \in \Fp[x]$ be a primitive  polynomial of degree $n$ and
$\{s_t\}$ be an m-sequence over $\Fp$ generated by the linear recurrence associated with $f(x)$.
Then the sequence $\{s_t\}$ satisfies the following properties:
\begin{enumerate}
	\item {\rm Span Property.} Each nonzero $n$-tuple over $\Fp$ occurs exactly once in  one period of $\{s_t\}$;
	 	\item {\rm Decimation Property.} For any positive integer $d$ coprime to $p^n-1$, the decimated sequence $\{s_{dt}\}$ is also an m-sequence of period $p^n-1$;
	\item {\rm Shift-and-Subtract Property.\footnote{This property is also referred to as \textit{Shift-and-And} property for binary m-sequences. Note that $\alpha^i - 1 $ (or $\alpha^i+1$ for $p=2$ ) for any $0<i\leq p^n-2$ is a nonzero element in $\Fpn$ and can be rewritten as $\alpha^j$ for certain integer $j$, where $\alpha$ is a primitive element of $\Fpn$.  }} For any shift $\tau\neq 0$ mod $p^n-1$,  the sequence $\{\widehat{s}_t\}$ with $\widehat{s}_t = s_{t+\tau} - s_t $ is also an m-sequence of period $p^n-1$;
	\item {\rm Balance Property.} In one period of $\{s_t\}$ each nonzero element in $\Fp$ occurs $p^{n-1}$ times and the zero element occurs $p^{n-1}-1$ times;
	\item {\rm Ideal Autocorrelation Property}. The autocorrelation function $C_s(\tau)$ of $\{s_t\}$ satisfies
	\[
	C_s(\tau) = \sum_{i=0}^{N-1}\omega^{s_{t+\tau} - s_{t}} =\begin{cases}
	p^n-1 & \text{ if } \tau \equiv 0 \mod{p^n-1},	
	\\
	-1 & \text{ if } \tau \not\equiv 0 \mod{p^n-1},
	\end{cases}
	\]where $\omega = \exp(2 \pi i/p)$ is a complex primitive $p$-th root of unity. 
	
\end{enumerate}
In particular, in the case of $p=2$, the binary m-sequence $\{s_t\}$ also satisfies:
\begin{enumerate}
	\addtocounter{enumi}{5}
	\item  {\rm Run Property.} In one period of $\{s_t\}$, 1/2 of the runs have
	length 1, 1/4 have length 2, 1/8 have length 3, 1/16 have
	length 4, and so on, as long as these fractions give integral
	numbers of runs, where a run of length $k$ is a block $s_{i}\ldots s_{i+k-1}$ such that $s_{i-1}\neq s_i=\cdots=s_{i+k-1}\neq s_{i+k}$. 
\end{enumerate}

\end{lemma}
From the linearity of the absolute trace function over $\Fpn$, it follows that $\Tr(ax)+\Tr(bx) = \Tr((a+b)x)$ and for any $c\neq 0$ in $\Fpn$,
\[
\sum_{x\in \Fpn} \omega^{\Tr(cx)}=\sum_{x\in \Fpn} \omega^{\Tr(x)} = 1+\sum_{i=0}^{p^n-2} \omega^{\Tr(\alpha^i)} = 0.
\] 
The first property in \thref{Lem1} is straightforward since $\{s_t\}$ is a linear recursive sequence with period $p^n-1$. 
Properties (2)-(5) listed in \thref{Lem1} can be proved by utilizing the trace expression of an m-sequence and its linearity property. The proof of Property (6) is a bit complicated and requires a closer investigation of the pattern of all nonzero $n$-tuple binary strings in an m-sequence \cite[Sec. 4.2]{Golomb1967}.  Golomb recognized the importance of
Properties (1), (4) and (6) in the context of randomness and called binary sequences satisfying these three properties pseudo-noise sequence \cite{Golomb1967}, where \textit{noise} is the standard designation in electronics for a random signal. Pseudo-noise sequences turned out to have many practical applications in digital communications, including  navigation, radar,
and spread-spectrum communication systems \cite{Golomb2005}.

	Let $\{u_t\}$ and $\{v_t\}$ be two $p$-ary sequences of period $N$. The (periodic) \textit{crosscorrelation} between them at shift $\tau$, where $0\leq \tau <N$, is given by
	\begin{equation}\label{Eq:PeriodicCor}
	C_{u, v}(\tau) = \sum_{i=0}^{N-1}\omega^{u_{t+\tau} - v_{t}}, 
	\end{equation} where $\omega = \exp(2 \pi i/p)$ is a complex primitive $p$-th root of unity. 
	When the two sequences are the same, we speak of the \textit{autocorrelation} of the sequence. 

Let $\alpha$ be a primitive element of $\Fpn$ and $\{s_t\}$ be an m-sequence given by $s_t=\Tr(\alpha^t)$ for $t=0, 1,\ldots, p^n-2$. Other $p$-ary m-sequences with period $p^n-1$ can be derived from $\{s_t\}$ by cyclic shifts and by taking its $d$-decimations with $\gcd(d, p^n-1)=1$ as a new sequence
$\{s_{dt}\}$, which corresponds to another $n$-th degree primitive polynomial in $\Fp[x]$.
The crosscorrelation between an m-sequence and its cyclic shift m-sequence is trivial, which is the autocorrelation of the m-sequence and has the ideal two-level property as in \thref{Lem1} (5).
Of more interest is the crosscorrelation between an m-sequence and its $d$-decimation sequence. This chapter concentrates on the cases where the decimations $d$ are coprime to the period $p^n-1$.
\begin{definition}[Crosscorrelation between m-sequences]
	Let $\{s_t\}$ be an m-sequence over $\Fp$ with period $p^n-1$ and $d$ a positive integer coprime to $p^n-1$. The crosscorrelation between the $\{s_t\}$  and its relative $d$-decimation $\{s_{dt}\}$ 
at a shift $\tau$ is given by
	\begin{equation}\label{Eq-WeiSum}
	C_d(\tau) = \sum\limits_{t=0}^{p^n-2} \omega^{s_{t+\tau}-s_{dt}}.
	\end{equation} 
	The crosscorrelation spectrum is defined as the set of values $$\{C_d(\tau)\,|\,\tau=0, 1,\dots, p^n-2\}.$$
	If the spectrum of $C_d(\tau)$ contains $t$ distinct values, we call the integer $d$ a $t$-valued decimation.
\end{definition}
From the properties of the trace function, some properties of the crosscorrelation $C_d(\tau)$ were recognized in \cite{Trachtenberg} and \cite{Helleseth1976}.
\begin{lemma}\thlabel{Lem2}\cite{Trachtenberg,Helleseth1976} 
	Let $C_d(\tau)$ be defined as in \eqref{Eq-WeiSum}. Then,
	\begin{enumerate}[label=i)]
		\item $C_d(\tau)$ is a real number;
		\item let $d\cdot d^{-1} \equiv 1 \mod{p^n-1}$, then the values and occurrences of each value of $C_d(\tau)$ and $C_{d^{-1}}(\tau)$ are the same;
		\item $C_{dp^j}(\tau) =C_{d}(p^j\tau) =C_{d}(\tau)$;
		\item $\sum\limits_{\tau=0}^{p^n-2}C_d(\tau) = 1$;
		\item $\sum\limits_{\tau=0}^{p^n-2}C_d(\tau-t)C_d(\tau) = \begin{cases}
		p^{2n}-p^n-1 \text{ when } t\equiv 0 \mod{p^n-1},\\
		-p^n-1 \text{ when } t\not\equiv 0 \mod{p^n-1};
		\end{cases}$
		\item $\sum\limits_{\tau=0}^{p^n-2}C^l_d(\tau) = -(p^n-1)^{l-1}+2(-1)^{l-1}+b_lp^{2n}$, where $b_l$ is the number of solutions of the equations 
		\begin{equation}\label{Eq-powermoment}
		\begin{array}{l}
		x_1 + x_2 + \cdots + x_{l-1} + 1 = 0, \\
		x_1^d + x_2^d + \cdots + x_{l-1}^d + 1 = 0,
		\end{array}
		\end{equation}where $x_i$ is a nonzero element in $\Fpn$ for $i=1, 2, \dots, l-1$.
	\end{enumerate}	
\end{lemma}
The above lemma is useful to determine the number of occurrences of each value $C_d(\tau)$ when there are only a few possible values in the crosscorrelation spectrum.

The crosscorrelation function between $m$-sequences has a close connection to the Walsh transform of the function $\Tr(x^d)$ over $\Fpn$.
\begin{definition}[Walsh Transform]\label{Def-WT}
	Let $f$ be a function from $\Fpn$ to $\Fp$. The Walsh transform of $f$ is the function $W_f: \Fpn \mapsto \mathbb{C}$ given by 
	\begin{equation}\label{Eq-WT}
	 W_f( a) = \sum_{x\in \Fpn}{\omega^{f(x)-\Tr(ax)}}.
	\end{equation}The Walsh spectrum of $f$ is the set of values $\{W_f(a)\,|\, a\in \Fpn\}.$ 
\end{definition}
The notion of Walsh transform can be also defined for a function from $\Fpn$ to itself as follows.
\begin{definition}\label{Def-WT2}
	Let $F$ be a function from $\Fpn$ to itself. The Walsh transform of $F$ is the function $W_F: \Fpn^2\mapsto \mathbb{C}$ given by 
	\begin{equation*}\label{Eq-WT}
	W_F(b, a) = \sum_{x\in \Fpn}{\omega^{\Tr(bF(x)-ax)}}.
	\end{equation*}The Walsh spectrum of $F$ is the set of values $\{W_F(b, a)\,|\, b \in \Fpn^*, a\in \Fpn\}.$ 
\end{definition}
From the above definitions, when $F(x)=x^d$ and $f(x) = \Tr(x^d)$ with $\gcd(d, p^n-1)=1$, one has 
$$
W_F(b, a) = W_F\left(1, \frac{a}{\sqrt[d]{b}}\right) = W_f\left(\frac{a}{\sqrt[d]{b}}\right)
$$for any $b\in \Fpn^*$ and $a\in \Fpn$; and when taking $a = \alpha^{\tau}$ for $0\leq \tau < p^n-1 $, one has
\begin{equation}\label{Eq-CC-WT}
C_d(\tau) = -1 + \sum_{z\in \Fpn} \omega^{\Tr(az-z^d)} = 
-1 + \sum_{x\in \Fpn} \omega^{\Tr(x^d-ax)} = -1 + W_f(a), 
\end{equation} where the second identity holds since $(-1)^d $ is equal to $1$ for both even and odd prime $p$
when $\gcd(d, p^n-1)=1$.
With the above identity, the power moments in \thref{Lem2} can be translated to those for Walsh transforms as follows (see, e.g.,  \cite[Corollary 7.4]{Katz-2012}).
\begin{lemma}\cite{Katz-2012}\thlabel{Lem2-B}
	Let $d$ be a positive integer coprime to $p^n-1$ and $W_d(a)$ denote the Walsh transform of $f(x)=\Tr(x^d)$ as given in \eqref{Eq-WT}. 
	Then for a positive integer $l$, the $l$th power moment of the Walsh transform satisfies
	$$
	P_d^{(l)} = \sum_{a\in \Fpn}W_d(a)^l = \frac{p^{2n}N^{(l)}-p^{ln}}{p^n-1},
	$$ where $N^{(l)}$ is the number of solutions of $(x_1, \ldots, x_l)\in \Fpn^l$ to the system of equations 
	\[
	\begin{array}{l}
	x_1 + x_2 + \cdots + x_{l} = 0, \\
	x_1^d + x_2^d + \cdots + x_{l}^d = 0.
	\end{array}
	\]
	In particular, for small integers $l=1, 2, 3, 4$, we have 
	\begin{enumerate}
		\item $P_d^{(1)} = \sum\limits_{a\in \Fpn}W_d(a)^1 = p^n$;
		\item $P_d^{(2)} = \sum\limits_{a\in \Fpn}W_d(a)^2 = p^{2n}$;
		\item $P_d^{(3)} = \sum\limits_{a\in \Fpn}W_d(a)^3 = p^{2n}M_1$;
		\item $P_d^{(4)} = \sum\limits_{a\in \Fpn}W_d(a)^4 = p^{2n}\sum\limits_{a\in\Fpn}M_a^2$,
	\end{enumerate}where $M_a= \#\{x\in \Fpn\,|\, (x+1)^d-x^d=a\}$.
\end{lemma}

\begin{remark}
	In the above lemma, we include the case $a=0$ in the summation of the power moments, which are identical to the ones in \cite[Corollary 7.4]{Katz-2012} since
	the Walsh transform of $\Tr(x^d)$ vanishes for any invertible exponent $d$ (Here ``invertible" refers to the property $\gcd(d,p^n-1)=1$). As we shall see in the subsequent sections, the power moments in \thref{Lem2} (resp. \thref{Lem2-B}) together with the 
	trivial ($0$-th) moment  $$\sum_{\tau=0}^{p^n-2}C_d(\tau)^{(0)}= p^n-1, \quad \text{resp. } \sum_{a\in \Fpn^*}W_f(a)^{(0)} = p^n-1$$ contribute several equations of the occurrences of each value of the crosscorrelation function (resp. Walsh transform) that help us the correlation (resp. Walsh) spectrum for an invertible exponent $d$. 
\end{remark}

The crosscorrelation function $C_d(\tau)$ also has a close connection to the weight distribution of cyclic codes. We briefly recall some basics of linear codes and cyclic codes for completeness. 
A linear $[\texttt{n}, \texttt{k}]$ code $\mathcal{C}$ over $\mathbb{F}_{p}$ is a $\texttt{k}$-dimensional subspace of $\mathbb{F}_{p}^{\texttt{n}}$. 
Let $A_i$ denote the number of codewords in $\mathcal{C}$ with Hamming weight $i$, where the Hamming weight of a codeword is the number of nonzero coordinates in the codeword.
The weight distribution of $\mathcal{C}$ is denoted by $(A_0,A_1,\cdots,A_{\texttt{n}})$. A linear code $\mathcal{C}$ over $\mathbb{F}_{p}$ is called {\em cyclic} if any cyclic shift of a codeword
is another codeword of $\mathcal{C}$.
It is well known that any cyclic code of length $\texttt{n}$ over $\mathbb{F}_{p}$ corresponds to an ideal of the polynomial residue class ring $\mathbb{F}_{p}[x]/(x^{\texttt{n}}-1)$ and can be expressed as $\mathcal{C}=\langle g(x) \rangle$, where $g(x)$ 
is monic and has the least degree. This polynomial is called the {\em generator polynomial} and $h(x)=(x^{\texttt{n}}-1)/g(x)$ is referred to as the {\em parity-check polynomial} of $\mathcal{C}$. 

\smallskip 

Consider the $p$-ary cyclic code $\mathcal{C}$ of length $p^n-1$ with parity-check polynomial $h(x)=m_{-1}(x)m_{-d}(x)$, where $m_i(x)$ is the minimal polynomial of $\alpha^{i}$ over $\Fp$. The Hamming weight of a nonzero codeword $\mathbf{c}_{a,b}$ in $\mathcal{C}$ can be expressed as
\begin{equation}\label{Eq-Code}
\begin{split}
wt(\mathbf{c}_{a, b}) & = p^n-1-\{ x\in \Fpn^*\,|\, \Tr(ax + bx^d) = 0\} \\
& = p^n-1 - \frac{1}{p}\sum_{x\in \Fpn^*}\sum_{y\in \Fp}\omega^{y\Tr(ax + bx^d)} \\
&=  p^n-1- \frac{1}{p}\left(p^n-p+\sum_{y\in \Fp^*}\sum_{x\in \Fpn}\omega^{y\Tr(ax + bx^d)} \right) \\
&=  p^{n-1}(p-1)-\frac{1}{p}\sum_{y\in \Fp^*}\sum_{x\in \Fpn}\omega^{\Tr(yax + ybx^d)} \\
&=  p^{n-1}(p-1)- \frac{1}{p}\sum_{y\in \Fp^*}\sum_{x\in \Fpn}\omega^{\Tr(x^d + \frac{ya}{\sqrt[d]{yb}}x)}  \\
& = p^{n-1}(p-1)- \frac{1}{p}\sum_{y\in \Fp^*}W_d\left(-\frac{a}{\sqrt[d]{b}}y^{\frac{d-1}{d}}\right),
\end{split}
\end{equation}	where
$W_d(a)$ denotes the Walsh transform of $\Tr(x^d)$ at the point $a$.   In particular, if $d\equiv 1 \mod{p-1}$, then 
$y^{\frac{d-1}{d}}=1$ for any $y\in \mathbb{F}_p^*$ and $W_d\left(-\frac{a}{\sqrt[d]{b}}y^{\frac{d-1}{d}}\right) = W_d\left(-\frac{a}{\sqrt[d]{b}}\right)$, which is independent of the choice of $y$. Hence each codeword in $\mathcal{C}$ has weight of the form
${(p-1)(p^{n}-W_d(a))}\big/{p}$ for certain $a\in \Fpn$ \cite{Katz-2012}. 
This indicates that  the weight distribution of $p$-ary cyclic codes with nonzeros $\alpha$ and $\alpha^d$ for any invertible exponent $d$ satisfying $d\equiv 1 \mod{p-1}$ can be completely determined by the Walsh spectrum of the Boolean function $\Tr(x^d)$, or equivalently, can be determined by the crosscorrelation of the m-sequence $\{s_t\}$ and its $d$-decimated sequence $\{s_{dt}\}$.

In applications it is of particular interest to find positive integers $d$ that lead to Walsh spectra or crosscorrelation spectra with
a few nonzero small absolute values \cite{Niho, Helleseth1976, Helleseth-1978, Helleseth-Kumar, Canteaut-Charpin-Dobbertin, Hollmann-Xiang}, 
since functions with high nonlinearity are resistant to linear attacks and sequence pairs with low crosscorrelation are easily distinguishable. 

In the following sections we will review the known results on crosscorrelation spectrum with few values in the literature,  of which most also have crosscorrelation with small absolute values.
With the identities in \eqref{Eq-CC-WT} and \eqref{Eq-Code}, some known results were in fact derived in the context of coding theory \cite{Kasami1966,Canteaut-Charpin-Dobbertin,Hollmann-Xiang}.

\section{Three-valued crosscorrelation}

A decimation integer $d$ is said to be {\it degenerate over $\Fpn$} when it is a power of $p$ modulo $p^n-1$; in this case $\Tr(x^d)=\Tr(x)$  is linear and  the crosscorrelation spectrum between $\{s_t\}$ and $\{s_{dt}\}$ degenerates to the autocorrelation spectrum of an m-sequence, which takes on two values $-1$ and $p^n-1$.
When $d$ is not degenerate, Helleseth showed \cite{Helleseth1976} that the crosscorrelation spectrum for two $p$-ary m-sequences of length $p^n-1$ with decimation $d$
takes on at least three values. In this sense,  finding decimation integers $d$ that lead to three-valued crosscorrelation spectra is of major interest. 
Moreover, with the identities (4)-(6) in \thref{Lem2}, one can easily determine the occurrences of each of the three values. 
So far there have been eleven infinite classes of decimations in the literature that lead to three-valued crosscorrelation spectra.
In the following we recall those decimations and their correlation distributions.

\begin{theorem}\thlabel{Th-3valued-2}
	Let $p=2$ and $C_d(\tau)$ be the crosscorrelation function defined as in \eqref{Eq-WeiSum}. The known decimations having 
	three-valued crosscorrelation and their correlation distributions are given as follows:
	\begin{enumerate}
		\item let $d=2^k+1$  or $d=2^{2k}-2^k+1$, where $n/\gcd(n,k)$ is odd and let $e=\gcd(n,k)$, then $C_d(\tau)$ takes on three values \cite{Gold1968, Kasami1971}:
		\begin{equation}\label{eq-3value-dist}
		\begin{array}{rcl}
		-1+2^{(n+e)/2} & \quad\text{ occurs } \quad & 2^{n-e-1} +2^{(n-e-2)/2} \text{ times,} \\
		-1-2^{(n+e)/2} & \quad\text{ occurs } \quad& 2^{n-e-1} - 2^{(n-e-2)/2} \text{ times,} \\
		-1 & \quad\text{ occurs } \quad & 2^n-2^{n-e}-1 \text{ times;} 
		\end{array}
		\end{equation}
		\item  let $d=2^m+2^{(m+1)/2}+1$ or $d=2^{m+1}+3$, where $n=2m$ with $m$ odd,   then $C_d(\tau)$ takes on three values $-1, -1\pm 2^{m+1}$ and the distribution is given by \eqref{eq-3value-dist} with $e=2$ \cite{Cusick-Dobbertin};
		\item let $d=2^m+3$, where $n=2m+1$, then 
		 $C_d(\tau)$ takes on three values $-1, -1\pm 2^{m+1}$  and the distribution is given by \eqref{eq-3value-dist} with $e=1$ \cite{Canteaut-Charpin-Dobbertin,Hollmann-Xiang};
		\item let $d=2^{\frac{n-1}{2}}+2^{\frac{n-1}{2}}-1$ if $n\equiv 1\mod{4}$ or $d=2^{\frac{n-1}{2}}+2^{\frac{3n-1}{2}}-1$ if $n\equiv 3 \mod{4}$, then 
		 $C_d(\tau)$ takes on three values $-1, -1\pm 2^{m+1}$  and the distribution is given by \eqref{eq-3value-dist} with $e=1$\cite{Hollmann-Xiang}.
	\end{enumerate}
\end{theorem}
The first two classes in Case (1) are the well-known Gold and Kasami-Welch exponents, respectively, for which the crosscorrelation spectra were calculated based on some properties of the
exponential sums of quadratic functions \cite[Ch. 5]{Lidl1996}; the two classes of decimations in Case (2) were conjectured by Niho in his PhD dissertation \cite[p. 74]{Niho}. Cusick and Dobbertin proved the conjectures by
proving two relevant trinomials $x^{2^{(m+1)/2}+2}+x+a$ and $x^6+x+a$ have either none or exactly two solutions in $\mathbb{F}_{2^m}$; 
Case (3) is the long-standing Welch's conjecture (see Golomb \cite{Golomb1968}) and was proved by Canteaut, Charpin
and Dobbertin \cite{Canteaut-Charpin-Dobbertin} and Hollmann and Xiang \cite{Hollmann-Xiang} independently. 
Case (4) was also conjectured in Niho's thesis \cite{Niho}. In \cite{Hollmann-Xiang} the authors proposed a unified way to 
determine the values of $C_d(\tau)$ for Cases (1), (3) and (4).

\begin{theorem}\thlabel{Th-3valued-3} Let $p=3$ and $C_d(\tau)$ be the crosscorrelation function defined as in \eqref{Eq-WeiSum}. The known decimations having 
	three-valued crosscorrelation and their correlation distributions are given as follows.
	\begin{enumerate}
		\item let $d=2\cdot 3^m+1$, where $n=2m+1$, then $C_d(\tau)$ takes on three values \cite{Dobbertin2001}:
		\begin{equation}\label{Eq-3-3value}
		\begin{array}{rcl}
		-1+3^{m+1} & \text{ occurs } & \frac{1}{2}(3^{n-1}+3^m)\text{ times,} \\
		-1- 3^{m+1} & \text{ occurs } & \frac{1}{2}(3^{n-1}-3^m)\text{ times,} \\
		-1 & \text{ occurs } & 3^n-3^{n-1}-1 \text{ times;} 
		\end{array}
		\end{equation}		
		\item let $d=3^k+2$, where $n=2m+1$ and $n | (4k-1)$ \cite{Dobbertin2001,Katz2015a}, then $C_d(\tau)$ takes on three values $-1, -1\pm 3^{m+1} $ with distribution given by \eqref{Eq-3-3value}.
	\end{enumerate}

\end{theorem}

\thref{Th-3valued-3} Case (1) is known as the ternary Welch case. Dobbertin, Helleseth, Kumar and Martinsen \cite{Dobbertin2001}  proved this case by showing
that 
$$3^{m+1} | (C_d(\tau)+1) \text{ and }
\sum_{\tau=0}^{3^n-2}(C_d(\tau)+1)^4 = 3^{3n+1},$$ and they also recognized by experimental evidence Case 2,
which is given in the form that $d=2\cdot 3^r + 1$ with $n | (4r+1)$.
By \thref{Lem2} (3) the decimation is changed to $d=3^k+2$ with $k=n-r$.
The conjecture was proved by Katz and  Langevin \cite{Katz2015a}, where authors employed diverse methods involving trilinear forms, counting points on curves via multiplicative character sums, divisibility properties of Gauss sums, and graph theory.

For general odd prime, known decimations having three-valued crosscorrelation are
due to Trachtenberg \cite{Trachtenberg} for odd $n$ and the generalizations are due to Helleseth \cite{Helleseth1976}.

\begin{theorem}\thlabel{Th-3valued-p}
		Let $p$ be an odd prime and $C_d(\tau)$ be the crosscorrelation function defined as in \eqref{Eq-WeiSum}. Let $d=(p^{2k}+1)/2$ or $d=p^{2k}-p^k+1$, where $p>2$ and $n/\gcd(n,k)$ is odd , then $C_d(\tau)$ takes on three values \cite{Trachtenberg,Helleseth1976}:
				\begin{equation}\label{Eq-3-3value-p}
			\begin{array}{rcl}
			-1+p^{{(n+e)}/2} & \quad\text{ occurs } \quad& \frac{1}{2}(p^{n-e}+p^{(n-e)/2})\text{ times,} \\
			-1-p^{{(n+e)}/2} & \quad\text{ occurs } \quad& \frac{1}{2}(p^{n-e}-p^{(n-e)/2})\text{ times,} \\
			-1 & \quad\text{ occurs } \quad& p^n-p^{n-e}-1 \text{ times,} 
			\end{array}
			\end{equation}		
			where $e=\gcd(n, k)$.
\end{theorem}
It is extremely difficult to find new decimations that give rise to three-valued crosscorrelation spectra. 
It's conjectured, up to equivalence (by considering the cyclotomic coset and the inverse of $d$ modulo $p^n-1$), that  \thref{Th-3valued-2}, \thref{Th-3valued-3}, and \thref{Th-3valued-p} cover all such decimations.
As a matter of fact, when $n=2^i$ for some positive integer $i\geq 2$, there
are no known decimations having three-valued crosscorrelation. Helleseth conjectured these cases in 1970s \cite{Helleseth1976} and recently Katz \cite{Katz-2012}
resolved the conjecture for the binary case. The conjecture remains open for odd primes $p>3$.

\smallskip 

\noindent\textbf{Open Problem 1.} Show that $C_d(\tau)$ takes on at least four values when $n=2^i$ with $i\geq 2$ for odd primes $p$.

\smallskip 

\noindent\textbf{Open Problem 2.} Show that $C_d(\tau)$ takes on at least four values if $d$ is not included in the known classes of three-valued decimations and their equivalent ones up to the cyclotomic cosets and inverse modulo $p^n-1$.

\section{Four-valued crosscorrelation}

One of the main contributions leading to decimations with four-valued crosscorrelation is due to Niho. In his thesis \cite{Niho} Niho studied a class of decimations of the form $d=s(2^m-1)+1$ for m-sequences with period $2^{2m}-1$.
This type of decimations were later generalized to the odd prime $p$ case, namely, the decimation is defined as $d=s(p^m-1)+1$ over the finite field $\mathbb{F}_{p^{2m}}$. 
Note that the power functions $x^d$ with such decimations are linear over the subfield $\mathbb{F}_{p^{m}}$, which is a nice feature that leads to many interesting results.
The decimation of the form $d=s(p^m-1)+1$ was later referred to as \textit{Niho exponent} in the construction of low-correlated sequences, highly nonlinear functions and permutations.
In the context of crosscorrelation, the main idea of using Niho exponent $d=s(p^m-1)+1$ is to reduce the problem to computing the number of solutions of certain equations in a subset of $\Fpn$.
More concretely, the transformation and its proof are recalled below \cite{Rosendahl}.

\begin{lemma}\thlabel{Lem3}
	Let $n=2m$ and $d=s(p^m-1)+1$ with a nonnegative integer $s$. 
	Let $C_d({\tau})$ be the crosscorrelation of the m-sequence $\{s_t\}$ and its $d$-decimation $\{s_{dt}\}$, where $s_t=\Tr(\alpha^t)$ and 
	$\alpha$ is a primitive element of $\Fpn$. 
	Then for each $a=\alpha^{\tau} \in \Fpn^*$, one has 
	$$W_d(a)=C_d(\tau)+1=(N(a)-1)p^m,$$  where 
	$N(a)$ is the number of distinct solutions to the equation
	\begin{equation}\label{Eq-Niho}
	x^{2s-1}-a x^s- \bar{a}x^{s-1} + 1 = 0
	\end{equation}
	in the unit circle  $U = \{ x\in \Fpn \,|\, x\cdot \bar{x} = 1\}$, where $\bar{x}=x^{p^m}$ for any $x\in \Fpn$.
\end{lemma}
\begin{proof}
	Let $q=p^m$ and denote $Y=\{\alpha^0,\alpha^1,\ldots,\alpha^q\}$.
	Then each element of $\Fpn^*$ is uniquely represented as $h y$ for some $h \in  \Fq^*$ and $y \in Y$. Note that $h^d=h$ for every $h \in \Fpm$. Thus we can write the crosscorrelation as 
	\begin{align*}
		W_d(a)& =1+C_d(\tau) \\
		& = \sum_{x\in \Fpn} \omega^{\Tr(x^d - a x)} 
		\\& = 1+\sum_{h\in \Fpm^*}\sum_{y\in Y}\omega^{\Tr(h(y^d  - ay))} 
		\\& = -p^m + \sum_{h\in \Fpm}\sum_{y\in Y}\omega^{\Tr^m_1(h\Tr_m^n(y^d  - ay))} 
		\\& = -p^m + \sum_{y\in Y}\sum_{h\in \Fpm}\omega^{\Tr^m_1(h(y^d + \bar{y}^d  - ay - \bar{a}\bar{y}))} 
		\\&=  p^m(Z(a)-1),
	\end{align*} where $Z(a)$ is the number of distinct solutions of $y^d + \bar{y}^d  - ay - \bar{a}\bar{y}=0$ in $Y$.
	Note that $\bar{y}^d = y^{s\cdot (q^2-q) + q} = y^{-(s-1)(q-1)+1}$ for any $y\in Y$. Since $Z(a)$ counts the number of $y\in Y$ satisfying $y^d + y^{dq}  - ay - \bar{a}y^q=0$,
	which (by dividing by $y^{s(q-1)+1}$) is the same as the number of $y \in Y$ satisfying $$y^{-(2 s-1)(q-1)}-a y^{-s(q-1)} - \bar{a} y^{-(s-1)(q-1)} + 1 =0.$$
	Taking $x=y^{q-1}$, one easily sees that $Z(a)$ is the same as the number of distinct $x\in \{y^{q-1}\,|\, y \in Y\}$ satisfying the equation 
	$$x^{2s-1}-a x^s- \bar{a}x^{s-1} + 1 = 0.$$
That is to say, $Z(a)$ is equal to $N(a)$. The desired conclusion follows.
\end{proof}

The next theorem provides a list, in historical order, of all the known decimations that were shown to give four-valued crosscorrelation for binary m-sequences. An important observation is that 
all the results in (1)-(4) are covered by the last case (5), which is clearly a Niho-type decimation.

\begin{theorem}\thlabel{Th-4valued-2}
	Let $v_2(i)$ be the highest power of $2$ dividing the integer $i$ and let $n=2m$. For the binary case $p=2$, the cross-correlation $C_d(\tau)$ takes on four values for the following decimations:
	\begin{enumerate}
		\item $d=2(2^m-1)+1$ where $m$ is even \cite{Niho};
		\item let $d=(2^{m/2}+1)(2^m-1) + 2$ where $m$ is even\cite{Niho};
		\item $d=\frac{2^{(m+1)t}-1}{2^t-1}$ where $m$ is even and $0<t<m$ with $\gcd(t, n)=1$ \cite{Dobbertin1998};
		\item $d=\frac{2^m-1}{2^t-1}(2^m-1)+2$ where $2t$ divides $m$ \cite{Helleseth2005};
		\item $d=s(2^m-1)+1$ where $s\equiv 2^r(2^r\pm 1)^{-1} \mod{2^m+1}$, where $v_2(r)<v_2(m)$ \cite{Dobbertin-Felke-Helleseth-Rosendahl}.
	\end{enumerate}
	Furthermore, the correlation distribution for these decimations can be unified as follows:
	\begin{equation}\label{Eq-4valued-2-dist}
		\begin{array}{rcl}
		  -1 - 2^m & \quad\text{ occurs }\quad & \frac{2^{n+r_1-1}-2^{m+r_1-1}}{2^{r_1}+1} \text{ times, }\\
		  -1  & \quad\text{ occurs }\quad & 2^{n-r_1}-2^{m-r_1} \text{ times, }\\
		  -1 + 2^m & \quad\text{ occurs }\quad & \frac{2^{n+r_1-1}-2^n+2^{m+r_1-1}}{2^{r_1}-1} \text{ times, }\\
		  -1 + 2^{r_1+m} & \quad\text{ occurs }\quad & \frac{2^{n}-2^m}{2^{3r_1}-2^{r_1}} \text{ times, }\\
		\end{array}
	\end{equation}
	where $r_1=\gcd(r, m)$ and it takes on values $1, \frac{m}{2}, 1, t$  for Cases (1)-(4),  respectively.
\end{theorem}

\noindent\textbf{Sketch of Proof.} From \thref{Lem3} it suffices to investigate the equation \eqref{Eq-Niho} for integers $s$ in \thref{Th-4valued-2}. It is easily seen that Case (5) covers Cases (1)-(4) with $s=2^{r}(2^r-1)^{-1}$ with $r=1, \frac{m}{2}, t, t$, respectively. It suffices to consider 
the case $s\equiv 2^r(2^r\pm 1)^{-1} \mod{2^m+1}$, for which the equation $x^{2s-1}-a x^s- \bar{a}x^{s-1} + 1 = 0$ can be transformed (after substitution of $z=x^{2^r\pm 1}$) as 
$$
z^{2^r\pm 1} + a z^{2^r} + \bar{a}z^{\pm 1} + 1 = 0.
$$  This gives
$$
z^{2^r+1} + a z^{2^r} + \bar{a}z + 1 = 0 \quad\text{ resp. }  az^{2^r+1} + z^{2^r} + z + \bar{a}= 0.
$$
The above equation were shown to have $0, 1, 2$ or $2^{r_1}+1$ solutions \cite[Lem. 22]{Dobbertin-Felke-Helleseth-Rosendahl}, which leads to four possible values in the crosscorrelation spectrum. Furthermore, the equation $(x+1)^d+x^d+1=0$
is proved to have solutions exactly in $\mathbb{F}_{2^m}$. This, combined with the first three power moments \thref{Lem2-B} (1)-(3) leads to the stated correlation distribution. 

\smallskip

There are numerical results that give decimations with four values that are not
explained by this list. However, one believes that Case (5) in \thref{Th-4valued-2} contains all four-valued cases when $d$ is of Niho type. The authors of \cite{Dobbertin-Felke-Helleseth-Rosendahl} stated the following conjecture.

\smallskip 

\noindent\textbf{Open Problem 3.} Any binary decimation of Niho type $d=s(2^m-1)+1$, $n=2m$ with four-valued crosscorrelation is covered
in the fifth decimation in \thref{Th-4valued-2}.

In the nonbinary case, the Niho-type decimation with $s=2$ also led to four-valued crosscorrelation, as recognized by Helleseth \cite{Helleseth1976}. 
Such four-valued decimations seem to be quite rare. More than three decades later, Zhang, Li, Feng and Ge found another ternary case \cite{Zhang2014}. The condition $\gcd(3,r)=1$ on the parameter $r$ for the decimation in \cite{Zhang2014} was generalized to arbitrary odd integers $r$ by Xia, Helleseth and Wu \cite{Xia2014}.
Below we recall the only known families having four-valued crosscorrelation for odd prime $p$. 

\begin{theorem}\thlabel{Th-4valued-p}
	For nonbinary cases, the known decimations having four-valued crosscorrelation $C_d(\tau)$ and the correlation distribution is given as follows:
	\begin{enumerate}
		\item let $d=2p^m-1$, where $n=2m$ and $p^m\not\equiv 2 \mod{3}$, then  $C_d(\tau)$  takes on four values \cite{Helleseth1976}: 
		\[\begin{array}{rcl}
			-1 - p^m & \quad\text{ occurs }\quad & \frac{1}{3}(p^n-p^m) \text{ times, }\\
			-1  & \quad\text{ occurs }\quad & \frac{1}{2}(p^n-p^m-2) \text{ times, }\\
			-1 + p^m & \quad\text{ occurs }\quad & p^m\text{ times, }\\
			-1 + 2p^{m} & \quad\text{ occurs }\quad & \frac{1}{6}(p^n-p^m) \text{ times;}
		\end{array}
		\]
		\item let $d=3^k+2$ or $d=3^{2k}+2$, where $n=3k$ and $k$ is odd, then $C_d(\tau)$ takes on four values \cite{Zhang2014,Xia2014}:
		\[\begin{array}{rcl}
		-1  & \quad\text{ occurs }\quad & 2\cdot 3^{3r-1}+3^{2r-1}-3^r-1\text{ times, }\\
		-1 + 3^{2r}  & \quad\text{ occurs }\quad & 3^r\text{ times, }\\
		-1 + 3^{\frac{3r+1}{2}} & \quad\text{ occurs }\quad & \frac{1}{2}(3^{3r-1}-3^{2r-1})\text{ times, }\\
		-1 - 3^{\frac{3r+1}{2}} & \quad\text{ occurs }\quad & \frac{1}{2}(3^{3r-1}-3^{2r-1})\text{ times.}
		\end{array}
		\]
	\end{enumerate}
\end{theorem}

For Niho-type decimation $d=s(2^m-1)+1$, $n=2m$, it was proved by Charpin \cite{Charpin2004} for $p=2$ and Helleseth, Lahtonen and Rosendahl \cite{Helleseth2007} for odd prime $p$
that $W_d(a)$ for $a\in \Fpn^*$ takes at least four values. Ranto and Rosendahl further characterized the possible four values for $W_d(a)$ as follows: for the binary case $W_d(a)$ takes values from
$\{-2^m, 0, 2^m, 2^{m+j}\}$ for some integer $j>0$ with $2j|m$; for the nonbinary case $W_d(a)$ takes values from $\{-p^m, 0, p^{m}, 2\cdot p^m\}$.
The result confirms the {\bf --1 conjecture} of Helleseth \cite{Helleseth1976} for Niho-type decimations.

\smallskip
\noindent\textbf{Open Problem 4 ($-1$ Conjecture \cite{Helleseth1976}).} Show that $-1$ always occurs as a value of $C_d(\tau)$
for any $d$ with $\gcd(d, p^n-1)=1$ and $d\equiv 1 \mod{p-1}$.

The above problem can be also stated as: for any $d$ with $\gcd(d, p^n-1)=1$ and $d\equiv 1 \mod{p-1}$, the Walsh transform $W_d(a)$ always vanishes at certain $a\in\Fpn^*$. 
It can be reformulated as a result of the number of common
solutions of a special equation system \cite{Helleseth2014a}.

\begin{lemma}\thlabel{Lem4}
	Let $q=p^n$ and $\alpha$ be a primitive element in $\Fq$. Let $N$ denote the number of solutions $x_i\in \Fpn$ of the equations system:
	\begin{equation*}
	\begin{array}{rrrrrrrrrrrrr}
	 x_0 &+& \alpha x_1 &+& \alpha^2 x_2 &+ \cdots +& \alpha^{q-2} x_{q-2} &=& 0,\\
	 	 x_0^d &+& x_1^d &+& x_2^d &+ \cdots +& x_{q-2}^d &=& 0.
	\end{array}
	\end{equation*}
	Then the $-1$ conjecture holds if and only if $N=q^{q-3}$ for any integer $d$ satisfying $\gcd(d, q-1)=1$ and $d\equiv 1 \mod{p-1}$.
\end{lemma}
The conjecture was based on numerical evidence and has been confirmed for all the exponents over finite fields $\mathbb{F}_{2^n}$ with $n$ up to $25$.
In addition, for the special case $d=p^n-2$ for which the condition $d\equiv 1 \mod{p-1}$ holds $p=2$ or $3$,
the crosscorrelation $C_d(\tau)$ becomes the Kloosterman sum. With the theory of elliptic curve, the corresponding Walsh transform $W_d(a)$
has been shown to take all values divisible by $4$ for $p=2$ and divisible by $3$ for $p=3$, respectively, in the range
$[1-2\sqrt{q}, 1+2\sqrt{q}]$. This confirms the conjecture for the case $d=p^n-2$ when $p=2$ or $3$. 

It appears difficult to find new infinite classes of decimations $d$ that have four-valued crosscorrelation. On the other hand, as mentioned previously, there do exist some specific four-valued decimations 
not covered in \thref{Th-4valued-2} and \thref{Th-4valued-p}. It would be interesting to find new infinite classes of such decimations.

\smallskip
\noindent\textbf{Open Problem 5.} Find new infinite classes of decimations having four-valued crosscorrelation for any prime $p$.

\section{Five-valued crosscorrelation}

We now consider five-valued decimations. 
Determining decimations that have exactly five-valued crosscorrelation 
is more difficult, especially when the correlation distribution cannot be settled. 
By far only a few families of such decimations with known correlation distributions
have been found. Some decimations are shown to have at most five values in their correlation spectra, but their correlation distributions are not yet determined.
This section recalls the known five-valued decimations and some relevant techniques used to determine the values of the crosscorrelation functions in the literature.
As before we start with the binary case.

\begin{theorem}\thlabel{Th-bin5value-1} \cite[Th. 4.8]{Helleseth1976}
	Let $n=2m$ and $d=2^m+3$. Then $C_d(\tau)$ takes on five values:
	\[
	\begin{array}{rcl}
	-1-2^m  & \quad\text{occurs} \quad & 2^{2m-1}-2^{m-3}(2^m+(-1)^{m+1}+1)-2^{m-1} \text{ times},\\
	-1 & \quad\text{occurs} \quad& \frac{1}{3} (2^{m}(2^m+(-1)^{m+1}+1)+2^{m-1}-3) \text{ times}, \\
	-1+2^m & \quad\text{occurs} \quad& 2^{2m-1}-2^{m-2}(2^m+(-1)^{m+1}+1) \text{ times}, \\	
	-1+2\cdot 2^m & \quad\text{occurs} \quad& 2^{m-1} \text{ times}, \\
	-1+3\cdot 2^m & \quad\text{occurs}\quad & \frac{1}{3}(2^{m-3}(2^m+(-1)^{m+1}+1)-2^{m-1}) \text{ times}.
	\end{array}
	\]
\end{theorem}

\noindent\textbf{Sketch of Proof.} The decimation $d=2^m+3$ is of Niho type as $d = s(2^m-1)+1$ with $s\equiv \frac{1}{4} \mod{2^m+1}$.
Then the equation \eqref{Eq-Niho} becomes 
$
x^{-\frac{1}{2}}  + a x^{\frac{1}{4}} + \bar{a} x^{-\frac{3}{4}}   + 1 = 0
$, which is equivalent to 
$
a x^4 + x^3 + x + \bar{a} = 0.
$ This implies that $W_d(a)=C_d(\tau)+1$ takes on possibles values $(i-1)2^m$ for $i=0, 1, \cdots, 4$. 

Denote by $t_i$ the number of occurrence of $(i-1)2^m$ in $W_d(a)$ for $a\in \mathbb{F}_{2^n}^*$. 
Niho \cite{Niho} found that the value $t_3=2^{m-1}$. As in \thref{Lem1} (6), denote by $b_3$ the number of solutions $(x_1,x_2)\in \mathbb{F}_{2^n}^*\times \mathbb{F}_{2^n}^*$ of
\[
\begin{cases}
x_1 + x_2 = 1, \\
x_1^{2^m+3} + x_2^{2^m+3} = 1.
\end{cases}
\] From the above equations one can deduce $(x_1^{2^m}+x_1)(x_1^2+x_1+1)=0$. Note that $x^{2^m}+x=0$ iff $x\in \mathbb{F}_{2^m}$ and $x^{2}+x+1=0$ iff $x\in \mathbb{F}_{2^2} \setminus \mathbb{F}_2$.
Since $ \mathbb{F}_{2^2}\subset \mathbb{F}_{2^m}$ iff $m$ is even, it follows that $b_3=2^m+(-1)^{m+1}+1$. 
With $t_3=2^{m-1}$ and $b_3=2^m+(-1)^{m+1}+1$, one can obtain the variables $t_i$'s for $i=0, 1, \cdots, 4$ by
solving the following  system of linear equations derived from the first three power moments in \thref{Lem2-B} (1)-(3):
\[
\begin{pmatrix}
1 & 1 & 1 & 1 & 1 \\
-2^m & 0 & 2^m & 2^{m+1} & 3\cdot 2^m \\
2^{2m} & 0 & 2^{2m} & 2^{2m+2} & 9\cdot 2^{2m} \\
-2^{3m} & 0 & 2^{3m} & 2^{3m+3} & 27\cdot 2^{3m} \\
\end{pmatrix}
\begin{pmatrix}
t_0 \\ t_1 \\ t_2 \\ t_3 \\ t_4
\end{pmatrix}
= \begin{pmatrix}
2^{2m-1} \\ 2^{2m} \\ 2^{4m} \\ 2^{4m}b_3
\end{pmatrix}.
\]
The above proof sketch reflects a common method to calculate the distribution of five-valued crosscorrelation, where
in addition to three trivial identities in \thref{Lem2-B} (1)-(2), one usually considers $b_3$ and also seeks another equation on the variables by further characterizing one of the unknowns or investigating the calculation of $b_4$.

\begin{theorem}\thlabel{Th-bin5-valued-2}\cite{Dobbertin1998}
	Let $n=4r$ for odd $r$ and $d=2^{2r}+2^{r}+1$. Then $C_d(\tau)$ takes on five values:
		\[
	\begin{array}{rcl}
	-1& \quad\text{occurs} \quad& 2^{4r-1}-2^{3r-2} \text{ times}, \\
	-1\pm 2^{2r} & \quad\text{occurs} \quad& (2^{4r-1}+2^{3r-1})/3 \text{ times}, \\	
	-1\pm 2^{2r+1}  & \quad\text{occurs} \quad & (2^{4r-2}-2^{3r-3})/3 \pm 2^{2r-2} \text{ times}.
	\end{array}
	\]
\end{theorem}
Dobbertin \cite{Dobbertin1998} studied $C_d(\tau)$ by investigating the multi-variant expression of relevant polynomials 
over $\mathbb{F}_{2^r}$. We refer readers to \cite[Sec. 3]{Dobbertin1998} for more details for the calculation of the crosscorrelation function and its distribution.

Up to now we have (intentionally) assume that $C_d(\tau)$ is the crosscorrelation between two m-sequences $\{\Tr(\alpha^t)\}$ and
$\{\Tr(\alpha^{dt})\}$ with period $p^n-1$, where $\alpha$ is a primitive element of $\Fpn$ and $\gcd(d, p^n-1)=1$.
It is clear that when decimations of the form $d\equiv d_1/d_2 \mod{p^n-1}$ are essentially covered when one generically 
exhaust decimations that are coprime to $p^n-1$ and lead to crosscorrelation with few values. On the other hand, as obtaining an explicit expression of $d_2^{-1} \mod{p^n-1}$ is generally intractable, it is of interest to consider decimations of the form $d_1/d_2 \mod{p^n-1}$ directly. For instance, the Kasami-Welch decimation $p^{2k}-p^k+1$, when $n/\gcd(n, k)$ is odd,
 corresponds to $(p^{3k}+1)/(p^k+1)\mod{p^n-1}$. The next theorem recalls decimations of this type in the case of $p=2$.

\begin{theorem}\thlabel{Th-bin-5valued-3}
	Let $p=2$, $n$ be an odd integer and $d=(2^{l}+1)/(2^k+1) \mod{2^n-1}$, where $l$ and $k$ are distinct positive integers. 
	Then for 
	\begin{center}
	(i) $(l, k)=(2t, t)$;	\quad (ii) $(l, k)=(5t, t)$; \quad  (iii) $(l, k) = (5t, 3t)$,
	\end{center}
the crosscorrelation $C_d(\tau)$ takes on at most five values \cite{Yu2006,Johansen2009}:
$$
\left\{-1, -1\pm 2^{(n+e)/2}, -1 \pm 2^{(n+3e)/2}\right\}, \text{ where } e=\gcd(n, t).
$$
\end{theorem}
\thref{Th-bin-5valued-3} can be deduced from a theorem by Kasami \cite{Kasami1971} on the weight distribution of subcodes of the second order Reed-Muller codes. 

\begin{theorem}[Kasami \cite{Kasami1971}]\thlabel{Th-Kasami}
	For odd $n$, let $t$ and $u$ be positive integers with $1\leq t \leq \frac{n-1}{2}$ and $1\leq u \leq \frac{n}{2e}+1$ with $e=\gcd(n,t)$. Let $A_t(u)$ be a binary cyclic code of length $2^n-1$ whose generator polynomial is given by $g_a(x)=\prod_{i=0}^{u-1}m_{1+2^{ti}}(x)$, where $m_i(x)$ is a minimal polynomial of $\alpha^i$ and $\alpha$ is a primitive element of $\mathbb{F}_{2^n}$. Similarly, let $F_t(u)$ be a binary cyclic code of length $2^n-1$ 
	with its generator polynomial given by $g_f(x)=\prod_{i=0}^{u-1}m_{1+2^{t(2i+1)}}(x)$. Then the dual codes of $A_t(u)$ and $F_t(u)$, denoted by $A_t(u)^{\perp}$ and $F_t(u)^{\perp}$, respectively, have the same weight distribution as those of $A_e(u)^{\perp}$
	whose distinct weights are given by 
	$$
	\{
	0, 2^{n-1}, 2^{n-1}\pm 2^{(n-e)/2+ie-1} 
	\} \text{ for } 1\leq i\leq u-1.
	$$
\end{theorem} 
\textbf{Proof Sketch of \thref{Th-bin-5valued-3}.} Consider codes $C_2$, $C_5$ and $C_{5/3}$ defined by
\begin{equation*}\label{key}
\begin{split}
C_2 = \{\Tr(ax^{1+2^k}+ b x^{1+2^{2k}})\,|\, a, b\in \mathbb{F}_{2^n}\}, \\
C_5 = \{\Tr(ax^{1+2^k}+ b x^{1+2^{5k}})\,|\, a, b\in \mathbb{F}_{2^n}\}, \\
C_{5/3} = \{\Tr(ax^{1+2^{3k}}+ b x^{1+2^{5k}})\,|\, a, b\in \mathbb{F}_{2^n}\}.
\end{split}
\end{equation*}
Consider the code $A_k(3)$ whose generator polynomial $g_a(x)$ has zeros $2, 1+2^k, 1+2^{2k}$.
It is clear that $C_2$ is a subcode of the dual of $A_t(3)$. In addition, 
the codes $C_5, \, C_{5/3}$ are subcodes of the dual of $F_t(u)$ defined in \thref{Th-Kasami} for $u=3$, where $F_t(u)$ has zeros $2^{k(2i+1)}$, $i=0, 1, 2$. Therefore, 
weight distributions of $C_2$, $C_5$ and $C_{5/3}$ are determined by \thref{Th-Kasami}. From the relation of $C_d(\tau)$ and the Walsh transform of the power function $x^d$, it immediately follows that $C_d(\tau)$ 
has at most five values as given in \thref{Th-bin-5valued-3}. \hfill$\square$

The decimations in \thref{Th-bin-5valued-3} (ii) and (iii) for $k=1$ equal $11$ and $5/3$, respectively. Boston and McGuire \cite{Boston2010} showed that
the crosscorrelation spectrum for these two decimations are identical.
For the case (iii) in \thref{Th-bin-5valued-3}, the correlation distribution for $k=1$ was further studied intensively in \cite{Johansen2009}, 
where Johansen and Helleseth determined  the correlation distribution by investigating the first four power moments and transformed the calculation of the fourth power moment to three special exponential sums:
\begin{equation*}
\begin{split}
&K(a) = \sum_{a\in \mathbb{F}_{2^n}^*}(-1)^{	\Tr(x^{-1}+ax)}, \\
&C(b, a) = \sum_{a\in \mathbb{F}_{2^n}^*}(-1)^{	\Tr(bx^{3}+ax)}, \\
&G(b,a) = 	\sum_{a\in \mathbb{F}_{2^n}^*}(-1)^{	\Tr(bx^{3}+ax^{-1})},
\end{split}
\end{equation*} where  the Kloosterman sum $K(a)$ and the cubic sum $C(b, a)$ have been extensively studied with some known results.
The work of \cite{Johansen2009} was later generalized to any $k$ when $n$ is odd and $\gcd(k,n)=1$ \cite{Johansen2009a}.
  The correlation distribution for general $k$ with $\gcd(k,n)=1$ is believed to be the same as the case $k=1$ when the identities in the next problem hold.
  
\textbf{Open Problem 6.} Let $n$ be an odd integer and $k$ be a positive integer coprime to $n$. Show that
\[
\sum\limits_{x\in \mathbb{F}_{2^n}^*}(-1)^{\Tr(x^{2^k+1}+x^{-1})} = \sum\limits_{x\in \mathbb{F}_{2^n}^*}(-1)^{\Tr(x^{3}+x^{-1})}
\] and 
\[
\sum\limits_{x\in \mathbb{F}_{2^n}^*}(-1)^{\Tr(x+x^{-1})}  = \sum\limits_{v\in \mathbb{F}_{2^n}^*}(-1)^{\Tr\left(\frac{(v^{2^k}+1)v^{2^k}}{(v^{2^k}+v)^{2^k+1}}\right)}.
\]

Two more families of Niho-type exponents $d=s(2^m-1)+1$ for $s=3, \,4$ and even $m$ have been shown to give five-valued crosscorrelation in the binary case\cite{Dobbertin-Felke-Helleseth-Rosendahl,Helleseth2021}. 

\begin{proposition}\cite{Dobbertin-Felke-Helleseth-Rosendahl}
	Let $n=2m$ where  $m$ is even and satisfies $m\not\equiv 2 \mod{4}$ and $d=3(2^m-1)+1$. Then $C_d(\tau)$ takes on at most five values $\{-1+j\cdot 2^m\,|\,j=-1, 0, 1, 2, 4\}$.
\end{proposition}
For the exponent $d=3(2^m-1)+1$, we see that $\gcd(d, 2^n-1) = 1$ is equivalent to $\gcd(5, 2^m+1)=1$, which holds if and only if $m\not\equiv 2 \mod{4}$.
Dobbertin et. al in \cite[Corollary 12]{Dobbertin-Felke-Helleseth-Rosendahl} showed that the polynomial
$
x^{5}+a x^3+\bar{a}x^{2} + 1 
$ in \eqref{Eq-Niho} cannot have 4 roots in $U$, which leads to the above statement. The crosscorrelation distribution for this case was recently settled by Xia et al. \cite{Xia-Li-Zeng-Helleseth},
where the authors transformed the major calculation of the fourth moment to a problem of counting the number of pairwise distinct $z_1, z_2, z_3, z_4 \in U^{4}\setminus \{1\}$ satisfying 
\[
z_1 + z_2 + z_3 + z_4  + 1 = 0.
\]
Surprisingly, this has a close connection with the binary Zetterberg codes \cite{Moisio2007} and can be resolved by the known results of weight distribution of Zetterberg codes.
\begin{theorem}\cite[Theorem 2 (i)]{Xia-Li-Zeng-Helleseth}
	Let $n=2m$ where $m$ is even and satisfies $m\not\equiv 2 \mod{4}$ and $d=3(2^m-1)+1$. Then $C_d(\tau)$ takes on five values:
	\[
	\begin{array}{rcl}
	-1-2^m  & \quad\text{occurs} \quad & \frac{11\cdot 2^{2m}-2^m\tau_m-10\cdot 2^m+1}{30}\text{ times},\\
	-1 & \quad\text{occurs} \quad& \frac{3\cdot 2^{2m}+2^m\tau_m-4\cdot 2^m-9}{8} \text{ times}, \\
	-1+2^m & \quad\text{occurs} \quad& \frac{2^{2m}-2^m\tau_m+6\cdot 2^m+1}{6} \text{ times}, \\	
	-1+2\cdot 2^m & \quad\text{occurs} \quad& \frac{2^{2m}+2^m\tau_m-2\cdot 2^m-1}{12} \text{ times}, \\
	-1+4\cdot 2^m & \quad\text{occurs}\quad & \frac{2^{2m}-2^m\tau_m+1}{120} \text{ times},
	\end{array}
	\]where $\tau_m$ is defined by 
	\[
	\tau_m = \left(\frac{1+\sqrt{-15}}{4}\right)^m + \overline{\left(\frac{1+\sqrt{-15}}{4}\right)^m} = 2\cos\left(m\cdot \arctan(\sqrt{15})\right)
	\] and its value can be computed by the following linear
recurrence equation with $\tau_1=1/2$ and $\tau_2=-7/4$:
\[
\tau_{k+2}= \frac{1}{2}\tau_{k+1}-\tau_k \text{ for } k\geq 1.
\]
\end{theorem}

For the case of $d=4(2^m-1)+1$ with even $m$, Niho conjectured \cite[Conjecture 4-6(5)]{Niho} that $C_d(\tau)$ takes on at most five values. Very recently Helleseth, Katz and Li resolved this
conjecture \cite{Helleseth2021}.

\begin{theorem}\cite[Theorem 3 (i)]{Helleseth2021}
	Let $n=2m$ with even $m$ and $d=4(2^m-1)+1$. Then $C_d(\tau)$ takes on at most five values $\{-1+j\cdot 2^m\,|\,j=-1, 0, 1, 2, 4\}$.
\end{theorem}
In the case of $d=4(2^m-1)+1$, it is clear that the crosscorrelation $C_d(\tau)$ has the form of $-1 + j\cdot 2^m$, where $j+1$ is the number of roots of
$$g_a(x) = x^7+ax^4+\bar{a}x^3+1$$ in $U$. The major task of the above statement is to show $g_a(x)$ cannot have $4$, $6$ and $7$ roots in $U$.
The proof of this statement is complicated and we provide the basic idea of the proof here. Denote by $\bar{x}=x^{2^m}$ the conjugate of $x$ in $\mathbb{F}_{2^m}$.
The polynomial $g_a(x)$ has a special property of being self-conjugate-reciprocal, which implies that $g_a(x)= 0$ if and only if $g_a({1}/{\bar{x}}) = 0$ (This property holds for other integers $s$ as well).
Consider the mapping $\Pi:\,\mathbb{F}_{2^n}\rightarrow \mathbb{F}_{2^n}$ given by $\Pi(x) = 1/\bar{x}$. Since $g_a(x)$ is self-conjugate-reciprocal, the set of its roots in its splitting field is therefore closed under the mapping $\Pi$.
Furthermore, it is shown in \cite{Helleseth2021} that when $g_a(x)$ is separable, let $R$ be the set of its seven distinct roots in the algebraic closure $\overline{\mathbb{F}_2}$, then
\[
S=\sums{\{u, v\}\subset R \\ u\neq v}\frac{uv}{(u-v)^2} = 0.
\]
Let $t$ be the number of distinct $\Pi$-orbits in $R$, it follows that \cite[Proposition 26]{Helleseth2021}
\[
{\rm Tr}_m(S) = {|R|+1 \choose 2} + t \mod{2}.
\]
The above results indicate that $R$ is a union of an even number of $\Pi$-orbits.
Since under the mapping $\Pi$ any element in $U$ is mapped to itself, in other words, the $\Pi$-orbit for any $x\in U$ has size one. This implies that when $g_a(x)$ is separable, it cannot have 4, 6 or 7 roots in $U$.
Nevertheless, the correlation distribution of $C_d(\tau)$ cannot be determined in a similar manner to the case $s=3$. 

\textbf{Open Problem 7.} Let $n=2m$ with even $m$. Compute the crosscorrelation distribution of $C_d(\tau)$ for $d=4(2^m-1)+1$.

For the nonbinary case, a family of ternary Niho-type decimation and a family of decimations for general odd prime $p$
have been found in \cite{Xia2017} and \cite{Helleseth1976}, respectively.

\begin{theorem}\cite[Th. 2]{Xia2017}
	Let $n=2m$ and $d=3(3^m-1)+1$ with $m\not\equiv 2 \mod{4}$. Then $C_d(\tau)$ takes on five values:
		\[
	\begin{array}{rcl}
	-1-3^m  &	 \quad\text{occurs} \quad & \frac{	11\cdot 3^{2m}-16\cdot 3^m-(-1)^m3^m+6}{30} \text{ times},\\
    -1  &	 \quad\text{occurs} \quad & \frac{	3\cdot 	
    	3^{2m}+2\cdot 3^m+(-1)^m3^m-14}{8} \text{ times},\\
    -1+3^m  &	 \quad\text{occurs} \quad & \frac{	3^{2m}-(-1)^m3^m+6}{6} \text{ times},\\
    -1+2\cdot 	3^m  &	 \quad\text{occurs} \quad & \frac{	3^{2m}+4\cdot 3^m+(-1)^m3^m-6}{12} \text{ times},\\
    -1+4\cdot 	3^m  &	 \quad\text{occurs} \quad & \frac{	3^{2m}-6\cdot 3^m-(-1)^m3^m+6}{120} \text{ times}.
	\end{array}
	\]
\end{theorem}
The above ternary case was proved by
following the idea used in \cite{Xia-Li-Zeng-Helleseth}, where the authors aimed to investigate the distribution of $C_d(\tau)$ for
$d=3(p^m-1)+1$ for any odd prime $p$. They showed that the value distribution of $C_d(\tau)$ is dependent on two combinatorial problems related to the unit circle of $\mathbb{F}_{p^{2m}}$. However, their techniques for handling the problems seemed to only work for $p=3$. The work resolved Open Problem 2.2 proposed in \cite{Dobbertin2001}. New techniques are required to calculate the distribution of $C_d(\tau)$ for general odd prime.

\textbf{Open Problem 8.} Let $n=2m$ with $m\not\equiv 2 \mod{4}$. Determine the distribution of $C_d(\tau)$ for $d=3(p^m-1)+1$ for odd primes $p\geq 5$.

The only proved family of five-valued decimation for a general odd prime $p$ is due to Helleseth \cite{Helleseth1976} and is recalled below.

\begin{theorem}\cite[Theorem 4.10]{Helleseth1976}
	Let $p$ be an odd prime and $p^n\equiv 1 \mod{4}$. Let $d=\frac{1}{2}(p^n-1)+p^i$ with $0\leq i <n$. Then $C_d(\tau)$ takes on five values:
		\[
	\begin{array}{rcl}
	-1  & \quad\text{occurs} \quad & \frac{1}{2}(p^n-5)\text{ times},\\
	-1 \pm  p^{n/2} & \quad\text{occurs} \quad& \frac{1}{4}(p^n-1) \text{ times}, \\
	-1+\frac{1}{2}(p^n\pm p^{n/2}) & \quad\text{occurs} \quad & 1 \text{ time}.	
		\end{array}
	\]
\end{theorem}
The calculation of $C_d(\tau)$ for $d=\frac{1}{2}(p^n-1)+p^i$  involves a splitting technique: 
\begin{equation*}
 \begin{split}
 C_d(\tau) & = -1 + \sum_{x\in \Fpn}\omega^{\Tr(x^d -ax)} \\
 &= -1 + \frac{1}{2}\left(\sum_{x\in \Fpn}\omega^{\Tr(x^2(1-a))}+\sum_{x\in \Fpn}\omega^{\Tr(x^2(1+a)\gamma) }\right),
 \end{split}
\end{equation*} where $\gamma$ is a non-square in $\Fpn^*$, i.e., $\gamma^{(p^n-2)/2} = - 1$. More detailed discussion of the above two exponential sums yields the correlation distribution in the above theorem.

\section{Six-valued crosscorrelation}
So far five families of decimations for binary m-sequences and two families of decimations for nonbinary case have been found in the literature. 
This section recalls those decimations and the corresponding correlation distributions.

\begin{theorem}\cite[Theorem 3.1]{Helleseth-1978}\thlabel{Th-H-1978}
	Let $n=4m$ with even $m$ and $d=2^{2m}-2^m+1$. Then $C_d(\tau)$ takes on the following values
		\[
	\begin{array}{rcl}
	-1+2^m & \quad\text{occurs} \quad &  (2^{4m}-2^m)/3 \text{ times},\\
	-1  & \quad\text{occurs} \quad&  2^{4m-1}-2^{3m-1}+2^{2m-1}-2^{m-1}-2 \text{ times}, \\
	-1-2^{2m}& \quad\text{occurs} \quad &  2^{3m}-2^{2m}\text{ times},\\
	-1-2^{2m+1}& \quad\text{occurs} \quad &  (2^{4m}-3\cdot 2^{3m} + 3\cdot 2^{2m}-2^m)/6 \text{ times},\\
	-1+2^{{3m}}& \quad\text{occurs} \quad & 1 \text{ time},\\	
	-1+2^{2m}(2^m-1)& \quad\text{occurs} \quad & 2^m \text{ times}.	
	\end{array}
	\]
\end{theorem}
\noindent\textbf{Sketch of Proof.} Denote $d_1=2^m(2^{2m}-2^m+1)=(2^{2m}+1)(2^m-1)+1$. Then the distribution of $C_d(\tau)$
is the same as $C_{d_1}(\tau)$. Take $N=2^m+1$ and denote 
$$
C_0=\{ x^{2^m+1}\,|\, x\in \mathbb{F}_{2^n}^*\}, \quad C_1 = \mathbb{F}_{2^n}^*\setminus C_0, \quad \text{and}\quad C_{\infty} = \{0\}.
$$
Then 
\[
\sum_{x\in \mathbb{F}_{2^n}} (-1)^{\Tr(ax^N)} = \begin{cases}
2^{4m}, & \text{ if } a \in C_{\infty}, \\
-2^{3m},& \text{ if } a \in C_{0}, \\
2^{2m},& \text{ if } a \in C_{1}.
\end{cases}
\]
Note that $(d_1-1)N \equiv 0 \mod(2^n-1)$.
Hence
\begin{equation*}
\begin{split}
C_{d}(\tau) &= \sum_{t=0}^{2^n-2}(-1)^{\Tr(\alpha^{d_1t}-\alpha^{t+\tau})} \\
&= \sum\limits_{j=0}^{N-1}\sum\limits_{i=0}^{\frac{2^{n}-1}{N}-1}(-1)^{\Tr(\alpha^{d_1(iN+j)}-\alpha^{\tau+iN+j})} \\
&= \sum\limits_{j=0}^{N-1}\sum\limits_{i=0}^{\frac{2^{n}-1}{N}-1}(-1)^{\Tr(\alpha^{iN}(\alpha^{jd_1}-\alpha^{\tau+j}))} \\
&= -1 +\frac{1}{N}\sum\limits_{j=0}^{N-1}\sum\limits_{x\in \mathbb{F}_{2^n}}(-1)^{\Tr(x^N(\alpha^{jd_1}-\alpha^{\tau+j}))}\\
& = -1 + \frac{1}{2^m+1}(2^{4m}n_{\infty}(c)-2^{3m}n_0(c)+2^{2m}n_1(c)),
\end{split}
\end{equation*}
where $n_i(c)$ for $i=0, 1, \infty$ and $c=\alpha^{\tau}$ is defined by $$n_i(c) = \#\{j\,|\, \alpha^{jd_1}-\alpha^{\tau+j} \in C_i, \, 0\leq j <N\}.$$
Helleseth transformed the problem of counting $n_i(c)$ to that of studying the number of solutions to certain equations, which gives the following possibilities for $(n_{\infty}(c), n_0(c), n_1(c))$:
\[
\begin{array}{cccc}
n_{\infty}(c) &  n_0(c) &  n_1(c) & \text{Number of occurrences} \\
0 & 0 & 2^m+1 & S_0 \\
0 & 1 & 2^m     & S_1 \\
0 & 2 & 2^m-1     & S_2 \\
0 & 3 & 2^m-2     & S_3 \\
1 & 0 & 2^m     & 1 \\
1 & 1 & 2^m-1     & 2^m
\end{array}
\]
Consequently $C_d(\tau)$ takes on the following values
 	\[
 \begin{array}{rcll}
 C_d(\tau) & = & -1 + (1-r)2^{2m} &\text{ if } (n_{\infty}(c), n_0(c), n_1(c)) = (0, r, 2^m+1-r), \\
                 & = & -1 + 2^{3m} &\text{ if } (n_{\infty}(c), n_0(c), n_1(c)) = (1,0, 2^m), \\
                 & = & -1 + (2^m-1)2^{2m} &\text{ if } (n_{\infty}(c), n_0(c), n_1(c)) = (1,1, 2^m-1), \\
 \end{array}
 \]where $0\leq r\leq 3$. This gives the possible values of the crosscorrelation function. The distribution of $C_d(\tau)$ is determined by the occurrence $S_r$ of $C_d(\tau)=-1+(1-r)2^{2m}$ for $0\leq r\leq 3$,
 which can be obtained from the first three power moments in \thref{Lem2-B}.
 
 The next two families of decimations are of Niho type. When $n=2m$ with odd $m$ and $d=3(2^m-1)+1$, the polynomial $g_a(x)$ in \eqref{Eq-Niho} has the form $g_a(x)=x^5 + a x^3 + \bar{a}x^2 + 1$.
 Since $g_a(x)$ can have $i\in\{0,1,2,3,4,5\}$ possible solutions, the corresponding $C_d(\tau)$ takes on at most six values. The distribution of $C_d(\tau)$ has remained unsolved for more than 30 years. 
 A significant progress is due to Dobbertin et al. \cite{Dobbertin-Felke-Helleseth-Rosendahl}, where the authors utilized properties of Dickson polynomials and  expressed the occurrences of possible values of $C_d(\tau)$
 in terms of Kloosterman sums over finite fields.
 \begin{theorem}\cite[Theorem 2]{Dobbertin-Felke-Helleseth-Rosendahl}
 	Let $n=2m$ with odd $m$ and $d=3(2^m-1)+1$. Then $C_d(\tau)$ takes on the following six values:
 		\[
 	\begin{array}{rcl}
 	-1-2^m  & \quad\text{occurs} \quad & \frac{11\cdot 2^{2m}-24\cdot 2^m+R}{30}\text{ times},\\
 	-1 & \quad\text{occurs} \quad& \frac{9\cdot 2^{2m}+22\cdot 2^m -3R-20}{24} \text{ times}, \\
 	-1+2^m & \quad\text{occurs} \quad& \frac{2^{2m}-2\cdot 2^m+R-4}{6} \text{ times}, \\	
 	-1+2\cdot 2^m & \quad\text{occurs} \quad& \frac{2^{2m}-R+12}{12} \text{ times}, \\
 	-1 + 3\cdot 2^m & \quad \text{occurs} \quad& \frac{2^m-2}{3} \text{ times}, \\
 	-1+4\cdot 2^m & \quad\text{occurs}\quad & \frac{2^{2m}-14\cdot 2^m+R+20}{120} \text{ times},
 	\end{array}
 	 	\]where $R$ is defined by 
 	\[
 	R = \sum_{y\in \mathbb{F}_{2^m}\setminus \mathbb{F}_2}(-1)^{\Tr_m(1/y)}K\left(\frac{1}{y^3+y}\right)
 	\] and $K(a)$ denotes the Kloosterman sum $$K(a)=\sum_{x\in \mathbb{F}_{2^m}}(-1)^{\Tr_m(x^{2^m-2}+ax)}.$$
 \end{theorem}
  Although the Kloosterman sum has been well studied and there are some results characterizing the values of Kloosterman sum, the calculation of $R$ in the above theorem is still intractable. As introduced in the previous section, 
  by an interesting connection to the binary Zetterberg codes, Xia et. al \cite{Xia-Li-Zeng-Helleseth} completely resolved the distribution of $C_d(\tau)$, which for odd $m$ is given as follows.
\begin{theorem}\cite[Theorem 2 (ii)]{Xia-Li-Zeng-Helleseth}
	Let $n=2m$ with odd $m$ and $d=3(2^m-1)+1$. Then $C_d(\tau)$ takes on the following six values:
	\[
	\begin{array}{rcl}
	-1-2^m  & \quad\text{occurs} \quad & \frac{11\cdot 2^{2m}-2^m\tau_m-22\cdot 2^m+1}{30}\text{ times},\\
	-1 & \quad\text{occurs} \quad& \frac{9\cdot 2^{2m}+3\cdot 2^m\tau_m+16\cdot 2^m-23}{24} \text{ times}, \\
	-1+2^m & \quad\text{occurs} \quad& \frac{2^{2m}-2^m\tau_m - 3}{6} \text{ times}, \\	
	-1+2\cdot 2^m & \quad\text{occurs} \quad& \frac{2^{2m}+2^m\tau_m-2\cdot 2^m+11}{12} \text{ times}, \\
	 -1 + 3\cdot 2^m & \quad \text{occurs} \quad& \frac{2^m-2}{3} \text{ times}, \\
	-1+4\cdot 2^m & \quad\text{occurs}\quad & \frac{2^{2m}-2^m\tau_m-12\cdot 2^m+21}{120} \text{ times},
	\end{array}
	\]where $\tau_m$ can be computed by the linear
	recurrence equation $\tau_{k+2}= \frac{1}{2}\tau_{k+1}-\tau_k$ for $k\geq 1$
	with $\tau_1=1/2$ and $\tau_2=-7/4$.
\end{theorem}
From the above distributions of $C_d(\tau)$ for $d=3(2^{m}-1)+1$, the parameter $R$ can be characterized in terms of $\tau_m$ as: $R = -2^m\tau_m+2\cdot 2^{m}+1$.
It would be interesting to explore what this relation indicates for the characterization of the Kloosterman sum.

For the Niho-type exponent $d=4(2^m-1)+1$, when $m$ is odd, Helleseth, Katz and Li \cite{Helleseth2021} used the same technique as described in the previous section to prove that $C_d(\tau)$ takes on at most six values. 
Nevertheless, the distribution of $C_d(\tau)$ still remains open.

\begin{theorem}\cite[Theorem 3 (ii)]{Helleseth2021}
	Let $n=2m$ with odd $m$ and $d=4(2^m-1)+1$. Then $C_d(\tau)$ takes on at most six values from the set $\{-1 + j \cdot 2^m \,|\, j=-1,0,1,2,3,4\}.$
\end{theorem}

\textbf{Open Problem 9.} Let $n=2m$ with odd $m$. Determine the distribution of $C_d(\tau)$ for $d=4(2^m-1)+1$.

We now move on to the decimations with six-valued crosscorrelation in nonbinary cases. As a matter of fact, the following decimation is a union of both binary case and nonbinary cases.

\begin{theorem}\cite[Theorem 4.11]{Helleseth1976} 
	Let $p$ be a prime satisfying $p \equiv 2\mod{3}$ and $n=2m$. Let $d=\frac{1}{3}(p^n-1)+p^i$, where $0\leq i<n$ and $f=\frac{p^n-1}{3}p^i \not 
	\equiv 2 \mod{3}$. Then $C_d(\tau)$ takes on the following values 
		\[
	\begin{array}{rcl}
	-1 & \quad\text{occurs} \quad & \frac{1}{9}(4p^n+2(-1)^{m+1}p^{m}-29)\text{ times},\\
	-1+(-1)^{m+1}p^m & \quad\text{occurs} \quad & \frac{1}{9}(2p^{2m}+2(-1)^{m}p^m-4)\text{ times},\\
	-1+(-1)^{m}p^m & \quad\text{occurs} \quad & \frac{1}{27}(8p^{2m}+2(-1)^{m}p^m-10)\text{ times},\\
	-1+2(-1)^{m+1}p^m & \quad\text{occurs} \quad & \frac{1}{27}(p^{2m}+2(-1)^{m+1}p^m+1)\text{ times},\\
	-1+\frac{1}{3}(p^{2m}+2(-1)^{m}p^m) & \quad\text{occurs} \quad & 1 \text{ time},\\
	-1+\frac{1}{3}(p^{2m}+(-1)^{m+1}p^m) & \quad\text{occurs} \quad & 2 \text{ times}
	\end{array}
	\]
	when $f\equiv 0\mod{3}$; and takes on the following values 
		\[
	\begin{array}{rcl}
	-1 & \quad\text{occurs} \quad & \frac{1}{9}(4p^{2m}+2(-1)^{m+1}p^m-20)\text{ times},\\
	-1+(-1)^{m+1}p^m & \quad\text{occurs} \quad & \frac{1}{9}(2p^{2m}+2(-1)^{m}p^m-4)\text{ times},\\
	-1+(-1)^{m}p^m & \quad\text{occurs} \quad & \frac{1}{27}(8p^{2m}+2(-1)^{m}p^m-28)\text{ times},\\
	-1+2(-1)^{m+1}p^m & \quad\text{occurs} \quad & \frac{1}{27}(p^{2m}+2(-1)^{m+1}p^m-8)\text{ times},\\
	-1+\frac{1}{3}(p^{2m}+2(-1)^{m}p^m) & \quad\text{occurs} \quad & 2 \text{ times},\\
	-1+\frac{1}{3}(p^{2m}+4(-1)^{m+1}p^m) & \quad\text{occurs} \quad & 1 \text{ time}
	\end{array}
	\]
	when $f\equiv 1\mod{3}$.
\end{theorem}
The proof of the above theorem used a similar technique as in \thref{Th-H-1978}. Choose $N=3$, since $(d-p^i)N\equiv 0 \mod{p^n-1}$, it follows that 
\[
C_d(\tau)+1=\frac{1}{3}\sum_{j=0}^{2}\sum_{x\in \Fpn}\omega^{\Tr((\beta^jx^3)^d-\alpha^t (\beta^jx^3))} =\frac{1}{3}\sum_{j=0}^{2}\sum_{x\in \Fpn}\omega^{\Tr(x^3(\beta^{jd}-\alpha^t \beta^j))},
\]where $\alpha$ is a primitive element of $\Fpn$ and $\beta = \alpha^{(p^n-1)/3}$. Denote $A=(-1)^{m+1}p^m$.  It is easily checked that 
\[
\sum_{x\in \Fpn}\omega^{\Tr(ax^3)} = \begin{cases}
p^n & \text{ when } a =0, \\
A & \text{ when } a \in C_0, \\
-\frac{1}{2}A & \text{ when } a\not\in C_0,
\end{cases}
\]where $C_0$ is the set of cubes in $\Fpn^*$. More detailed analysis of the triple $(1-\alpha^t, \beta^{d}-\alpha^t \beta, \beta^{2d}-\alpha^t \beta^2)$ indicates that
$C_d(\tau)$ takes on the following values for $f\equiv 0\mod{3}$:
\[
-1, -1+\frac{1}{2}A, -1-\frac{1}{2}A, -1+A, -1+\frac{1}{3}(p^n-A), -1+\frac{1}{3}(p^n+\frac{1}{2}A)
\] and takes  on the following values for $f\equiv 1\mod{3}$: 
\[
-1, -1+\frac{1}{2}A, -1-\frac{1}{2}A, -1+A, -1+\frac{1}{3}(p^n-A), -1+\frac{1}{3}(p^n+2A).
\] 
The calculation of the distribution of $C_d(\tau)$ relies on the investigation of the first three power moments.

The following decimation is the $p$-ary variant of \thref{Th-H-1978}. It is 
also similar to the three-valued
decimation of the form $d=p^{2k}-p^k+1$ under the condition that $n/\gcd(n,k)$ is odd. As this condition is not satisfied, 
the crosscorrelation has six possible values. 

\begin{theorem}\cite[Theorem 1]{Helleseth2003}
	Let $n=4m$, $p^m\not\equiv 2 \mod{3}$ and let $d=p^{2m}-p^m+1$. Then $C_d(\tau)$ takes on the following values 
	\[
	\begin{array}{rcl}
	-1-2p^{2m} & \quad\text{occurs} \quad &  (p^{4m}-3p^{3m}+3p^{2m}-p^m)/6 \text{ times},\\
	-1-p^{2m} & \quad\text{occurs} \quad & p^{3m}-p^{2m}\text{ times},\\
	-1 & \quad\text{occurs} \quad &  (p^{4m}-p^{3m}+p^{2m}-p^m-4)/2 \text{ times},\\
	-1+p^{2m} & \quad\text{occurs} \quad &  (p^{4m}-p^m)/3 \text{ times},\\
	-1+p^{3m}-p^{2m} & \quad\text{occurs} \quad & p^m\text{ times},\\
	-1+p^{3m} & \quad\text{occurs} \quad & 1 \text{ time}.
	\end{array}
	\]
\end{theorem}
The proof adopted a similar technique as in \thref{Th-H-1978}. Take $N=p^m+1$. Firstly, the exponential sum
\[
\sum_{x\in \mathbb{F}_{p^{4m}}}\omega^{\Tr(ax^{N})} = \begin{cases}
p^{4m} & \text{ when } a \in C_{\infty}, \\
-p^{3m} & \text{ when } a \in C_0, \\
p^{2m} & \text{ when } a\in C_1,
\end{cases}
\]where $C_{\infty}=\{0\}$, $C_0=\{x^{N}\,|\,x \in \mathbb{F}_{p^{4m}}^*\}$ and $C_1 =  \mathbb{F}_{p^{4m}}^* \setminus C_0.
$ 
Denote $d_1=dp^m$. Since  $(d_1-1)N\equiv 0 \mod{p^{4m}-1}$, the crosscorrelation $C_d(\tau)$ can be expressed as
\[
\begin{split}
C_d(\tau) & =-1+\frac{1}{N}\sum_{j=0}^{N-1}\sum_{x\in \mathbb{F}_{p^{4m}}}\omega^{\Tr(x^{p^m+1}(\alpha^{dj}-\alpha^{\tau}\alpha^j))} \\
&=-1+\frac{1}{N}\left( 
p^{4m}n_{\infty}(c) -p^{3m}n_0(c) + p^{2m}n_1(c)
\right),
\end{split}
\]where $n_i(c)$ for $i=0, 1, \infty$ and $c=\alpha^{\tau}$ is defined by $$n_i(c) = \#\{j\,|\, \alpha^{jd_1}-\alpha^{\tau+j} \in C_i, \, 0\leq j <N\}.$$
As in the proof of \thref{Th-H-1978}, the remaining proof relies on the analysis of possible values of $(n_{\infty}(c),n_0(c),n_1(c))$ and 
 the number of solutions of the equations
\[\begin{split}
x_1 + x_2 + 1 &= 0,  \\
x_1^d + x_2^d + 1 &= 0,
\end{split}
\] which is shown to be $p^m$. The first three power moments yield the complete distribution of $C_d(\tau)$.

\section{Conclusion}
The crosscorrelation of m-sequences is a fascinating and challenging research topic of both theoretical and applied interest.
Recent years have witnessed some advances on the topic, meanwhile, many problems still remain open. This chapter provides an updated overview of the topic with the focus on
reviewing known and new techniques in the literature.

\bibliographystyle{plain}
\bibliography{sequence}


\end{document}